\documentclass[a4paper,11pt]{article}
\usepackage{graphicx} % Include figure files
\usepackage{amsmath, relsize} 
\usepackage{todonotes} 
\usepackage{amsthm} 
\usepackage{amsfonts}
 \usepackage{bbm}
\usepackage{amssymb}
\usepackage{array}
\usepackage{url}
\usepackage{amsmath}
\usepackage{amssymb}
\usepackage{amsthm}
\usepackage{enumerate}
\usepackage{setspace}
\usepackage{graphicx,color,hyperref}
\usepackage[margin=1in]{geometry}
\parskip \medskipamount
\newtheorem{theorem}{Theorem}
\newtheorem{proposition}[theorem]{Proposition}

\theoremstyle{definition}

\newtheorem{remark}[theorem]{Remark}

\newcommand{\lr}[1]{\left(#1\right)}

\title{Absorbing-state phase transition in biased activated random walk}
\author{Lorenzo Taggi \vspace{0.1cm} \\ \textit{Max Planck Institute for Mathematics in the Sciences, Leipzig, Germany}}
\date{}

\begin{document}
\maketitle

\section*{Abstract}
We consider the activated random walk (ARW) model on $\mathbb{Z}^d$,
which undergoes a transition from an absorbing regime to a regime
of sustained activity.
In any dimension we prove that
the system is in the active regime when the particle density is less than one, 
provided that the jump distribution is biased and that the sleeping rate is small enough.
This answers a question from \textit{Rolla and Sidoravicius} (2012) and 
\textit{Dickman, Rolla and Sidoravicius} (2010) in the case of biased jump distribution.
Furthermore, we prove that the critical density depends on the jump distribution.

\section{Introduction}
In this paper we consider the activated random walk (ARW) model on the lattice. 
This is a continuous-time interacting particle system with conserved number of particles, 
where each particle can be in one of two states: A (active) or S (inactive, sleeping). 
Each A-particle performs an independent, continuous time
random walk on $\mathbb{Z}^d$ with jump rate $1$
and jump distribution $p( \cdot )$.
Moreover, every A-particle has a Poisson clock with rate $\lambda>0$
(\textit{sleeping rate}).
When the clock rings,
if the particle does not share
the site with other particles,
the transition $A \rightarrow S$ occurs,
otherwise nothing happens.
S-particles do not move and remain sleeping until the instant
when an other particle is present at the same vertex.
At such an instant, the particle which is in the S-state flips
to the A-state, giving the transition A+S $\rightarrow$ 2A.
The initial particle configuration is distributed
according to a product of Bernoulli distributions
having expectation $\mu \in [0, 1]$,
that we call \textit{particle density}.
As we consider initial configurations with only active particles,
from the previous rules it follows that sleeping
particles can be observed only if they 
occupy the site alone.

In ARW a phase transition arises from a
conflict between the spread of the activity and 
a tendency of the activity to die out.
We say that ARW exhibits \textit{local fixation}
if for any finite set  $V \subset \mathbb{Z}^d$, 
there exists a finite time $t_V$ such that after 
this time the set $V$ contains no active particles.
We say that ARW \textit{stays active} if local fixation does not occur.

Some of the central 
questions for this model involve the estimation
of the  \textit{critical density} which separates the two regimes,
$$
\mu_c(\lambda, p(\, \cdot \, ) ) := \inf \, \{ \mu \in [0,1] \, : \, \mathbb{P}(\mbox{ARW is active} ) \, > \, 0 \},
$$
where $\mathbb{P}(\mbox{ARW is active} )$ is intended as 
a function of the parameter $\mu$.
The 0-1 law and the 
monotonicity properties that have been proved
in the seminal article by Rolla and Sidoravicius \cite{Rolla}
imply that if $\mu > \mu_c$, then ARW sustains activity \textit{almost surely}.

In several articles an estimation for $\mu_c$ has been provided.
In one dimension, it has been proved by Rolla and Sidoravicius 
\cite{Rolla} that $\mu_c \in [\frac{\lambda}{1+\lambda},1]$.
Our definition of $\mu_c$
 implies that $\mu_c\leq 1$ since particles are initially distributed as Bernoulli random variables. However, even if we replace this with any product measure of density $\mu>0$, 
it is intuitive that $\mu_c\leq 1$, since at most one particle can fall asleep at any given vertex.
This fact has been proved in \cite{Gurevich, Rolla, Shellef}
in wide generality.
A fundamental question for this model
is whether $\mu_c < 1$ for any sleeping rate $\lambda$.
This question has been asked by
Dickman, Rolla and Sidoravicius \cite{Dickman}
and by Rolla and Sidoravicius \cite{Rolla}
and its answer is expected to be positive in wide generality.
In this article we provide a positive answer to this question 
in any dimension
in the case of biased jump distribution.
In particular, in one dimension we prove a stronger
statement, i.e, that $\mu_c \to 0$ as $\lambda \to 0$.

We are now ready to state our results.
We let  $ \mathbf{m} = \sum_{z \in \mathbb{Z}^d} p(z) \, z$
be the expected jump of the random walk,
we let $\mathbf{e}_j$ be the axis direction such that 
$\mathbf{m} \cdot \mathbf{e}_i$ takes the maximum value,
we let $\mathcal{H} = \{ z \in \mathbb{Z}^d \, \mbox{ s.t. }  \mathbf{e}_{j} \cdot z \leq 0\}$
and we define the number,
\begin{equation}
\label{eq:defF}
F\left(\lambda, p(\, \cdot \,)\right) := E[(1 + \lambda)^{-\ell_{\mathcal{H}}}],
\end{equation}
where $\ell_{\mathcal{H}}$ is the total time spent on $\mathcal{H}$ by a discrete time random walk
with jump distribution $p(\, \cdot \,)$.
Such a number is the probability that a continuous time random walk never deactivates, if it jumps at rate $1$ 
and it deactivates at a rate $\lambda$ only when it is in $\mathcal{H}$.
As a consequence of the law of large numbers, for any jump distribution such that
$\mathbf{m} \neq \mathbf{0}$ and for any $\lambda >  0$, such a probability is positive
and, furthermore, $\lim\limits_{\lambda \rightarrow 0} F(\lambda, p(\, \cdot\, )) = 1$, as
the walker spends only a finite amount of time in $\mathcal{H}$.

\begin{theorem}
\label{theo:maintheo1}
Consider ARW on $\mathbb{Z}$ 
with jump distribution $p (\, \cdot \, )$
having a finite support and
such that $\mathbf{m} \neq \mathbf{0}$.
Then, $$\mu_c\left(\lambda, p(\, \cdot \, ) \right) \leq 1 - F\left(\lambda, p(\, \cdot\, )\right).$$
\end{theorem}
\begin{figure}
\begin{center}
\includegraphics[scale=0.75]{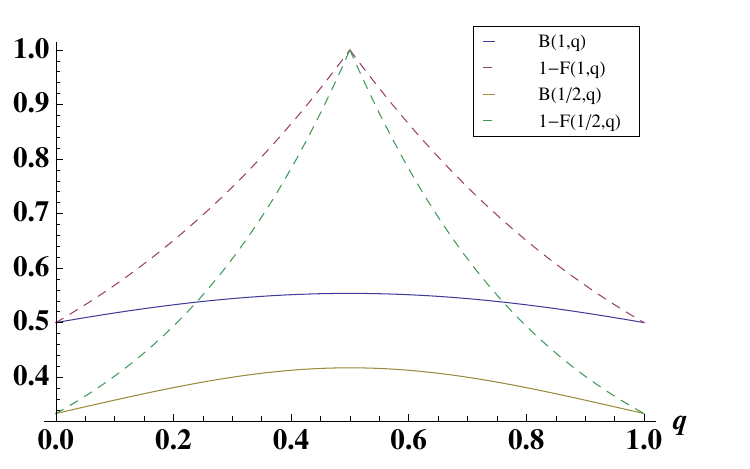}
\caption{Upper and lower bound (respectively, dashed and continuous lines) for the critical density in one dimension and jumps on nearest neighbours, $p(1)=q$ and $p(-1)=1-q$. }
\label{fig:plot}
\end{center}
\end{figure}

The next theorem provides an upper bound for the critical density
in dimension $d \geq 2$.
\begin{theorem}
\label{theo:maintheo2}
Consider ARW on $\mathbb{Z}^d$ 
with jump distribution $p (\, \cdot \, )$
having a finite support and such that $\mathbf{m} \neq \mathbf{0}$.
Then,
\begin{equation}
\mu_c\left (\lambda, p(\, \cdot \, ) \right) \leq \frac{1}{ {F(\lambda, p(\, \cdot \, ) )} + 1}.
\end{equation}
\end{theorem}

Although $\mu_c$ is conjectured to be strictly less than 
one for any positive $\lambda$ and for any jump distribution,
our proof techniques allow to answer 
such a question only under the assumption
of biased jump distribution.
A second, natural question is how and whether
the critical density depends on the jump distribution.
Our third theorem states that the critical
density is not a constant function of the jump distribution.
\begin{theorem}
\label{theo:depend}
Consider ARW with jump distribution on nearest neighbours, $p(1)=q$ and $p(-1)=1-q$, where $q \in [0,1]$.
For any fixed  $\lambda \in \mathbb{R}_+$, the critical density $\mu_c(\lambda, q)$ is not a constant function of $q$.
\end{theorem}
The proof of the theorem uses the stabilization
procedure of
Rolla and Sidoravicius \cite{Rolla}
and it is based on an observation.
In particular, we provide a new lower bound for the critical density as a function
of the sleeping rate and of the bias parameter (see Figures \ref{fig:plot} and \ref{fig:coparison})
and we prove that  $\mu_c(\lambda, q) > \frac{\lambda}{1+\lambda}$
when $q \not\in \{0, 1\}$.
The statement of Theorem \ref{theo:depend} follows from our lower bound,
as it is known \cite{Hoffman} that $\mu_c(\lambda, q) =\frac{\lambda}{1+\lambda}$
when $q \in \{0, 1\}$.
\begin{figure}
\begin{center}
\includegraphics[scale=0.7]{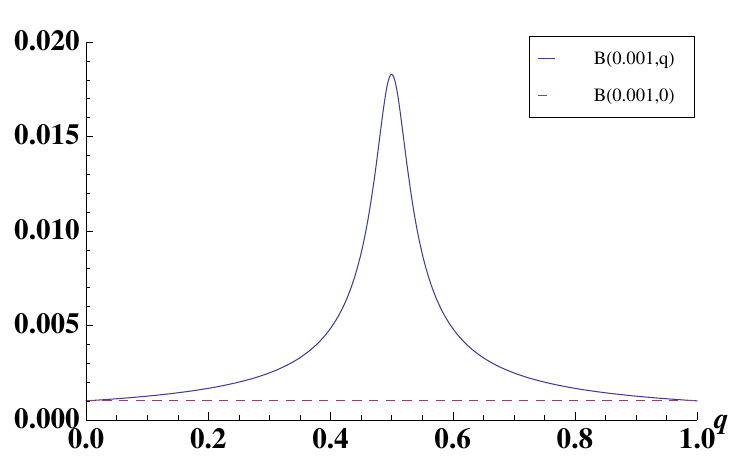}
\caption{Lower bound $B(q,\lambda)$  for the critical density for low lambda ($\lambda=1/1000$) as a function of the bias parameter $q$ (continuous line), contrasted with the lower bound $\lambda / (1+\lambda)$ from~\cite{Rolla} (dashed horizontal line).}
\label{fig:coparison}
\end{center}
\end{figure}

\begin{remark}
Our Theorems \ref{theo:maintheo1} and \ref{theo:depend} 
hold for any distribution
 of the initial location of the particles which is a product
 of identical distributions parametrized by their expectation $\mu$. 
 On the contrary, if we fixed beforehand 
 a distribution which is different from Bernoulli,
 the statement of Theorem \ref{theo:maintheo2} would 
 be that $\mu_c < 1$ only for small enough $\lambda$.
\end{remark}

We end this introductory section by presenting the structure of the article.
In Section \ref{sect:description} we introduce the proofs of Theorems
\ref{theo:maintheo1} and \ref{theo:maintheo2} to the reader.
In Section \ref{sect:Diaconis} we present the Diaconis-Fulton graphical representation, 
which is a fundamental framework for the analysis of ARW.
In Section \ref{sect:upperbound1} we prove our upper bound for the critical density in one dimension.
In Section \ref{sect:upperbound2} we prove our upper bound in more than two dimensions.
In Section \ref{sect:lowerbound} we sketch
the stabilization algorithm of Rolla and Sidoravicius 
and we present our observation for the proof of Theorem \ref{theo:depend}.

\section{Some words on the proofs}
\label{sect:description}
Our proofs rely on the
discrete Diaconis-Fulton representation
for the dynamics of ARW.
As it has been proved in~\cite{Rolla},
local fixation for ARW
is related to the stability properties of this 
representation, which leaves aside the chronological
order of events. 

At every site $x \in \mathbb{Z}^{d}$, an infinite sequence of independent and 
identically distributed random variables is defined.
Their outcomes are some operators (``instructions'') acting on the current particle configuration by moving one particle from one site to the other one or by trying to let the particle turn to the S-state.

Local fixation for the dynamics of ARW is related to the
the number of instructions that must be used
in order to stabilize the initial  particle configuration.
Denote by $B_L$ a compact subset
of $\mathbb{Z}^d$ such that $B_L \uparrow \mathbb{Z}^d$
as $L \rightarrow \infty$. For every $x \in \mathbb{Z}^d$,
let $m_{B_L, \eta, \tau} (x)$ be
the number of instructions that must be used at $x$
in order to make the configuration $\eta$
stable in  $B_L$ according to the instructions $\tau$ and denote by $\xi_{B_L, \eta, \tau}$
the corresponding stable configuration.
A configuration is stable in $B_L$ if there are no active particles in $B_L$. 
A fundamental property of
the representation is 
\textit{commutativity}, i.e.,
$\xi_{B_L, \eta, \tau}$
and $m_{B_L, \eta, \tau}$
do not depend on the order
according to which instructions have been used.
A second property of the representation 
is that if there exists a positive constant 
$c$ such that for every integer $L$ large enough, 
\begin{equation}
\label{eq:implication}
\mathcal{P}^{\nu} (m_{B_L, \eta, \tau} ( 0 )  =0  ) \geq c,
\end{equation}
then ARW fixates almost surely.
Analogously, if there exists a positive constant $c^{\prime}$
such that for every integer $L$ large enough,
\begin{equation}
\label{eq:implication2}
\mathcal{P}^{\nu} (m_{B_L, \eta, \tau} ( x )  > c^{\prime} \, L ) \geq c^{\prime},
\end{equation}
then ARW stays active almost surely.
The proof of our results is based on the definition of stabilization algorithms 
for the set $B_L$ and on counting the number of particles
crossing the origin, which is chosen to belong to the inner boundary of $B_L$.
In order to prove the upper bound (resp. the lower bound),
we provide an estimation of the choice of parameters such that 
$(\ref{eq:implication})$ (resp. \ref{eq:implication2}) holds
for every $L$ large enough.

The proof of Theorem \ref{theo:maintheo2}
is based on the following idea. 
In two dimensions, we introduce the set $B_L = [-L+1, 0] \times [-L^3, L^3 ] $ by assuming that
$\mathbf{m} \cdot \mathbf{e}_1 >0$ by symmetry.
We define a stabilization procedure where particles are moved one by one
until a certain ``stopping'' event occurs. 
By ``moving'', we mean that we use always the instruction
on the site where the particle is located until such an event occurs.
We say that a particle is ``good'' if it stops on 
one of the sites which is empty in the initial particle configuration or if it leaves
$B_L$ from the boundary side containing the origin.
Because of the choice of our stopping events, of the order according
to which particles are moved and of the bias of the jump distribution,
we can provide a positive uniform
lower bound $F$ for the probability that a particle is good.
Thus, we show that, if the density of good particles $\mu \cdot F$ is
higher than the density of empty sites $1 - \mu$, then a positive
density of particles leaves $B_L$ by crossing the boundary side containing the origin.
In one dimension this would be enough to prove almost sure activity when $\mu< 1$
with $B_L = [-L,0]$, as the number of sites belonging to the inner boundary of $B_L$ 
does not grow to infinity with $L$.
Instead, in two or more dimensions 
a control of which boundary sites are crossed by the particles
jumping away from $B_L$ is needed. To obtain such a  control,  we adapt to our setting
the method of ghost explorers ~\cite{Lawler} and we exploit the symmetry properties of the random walk. 
Thus, we prove that the number of particles crossing the origin
before leaving $B_L$ is larger than $c  L$ for some $c>0$ with high probability.

This idea applies also to the one dimensional case, but actually the stabilization procedure that has been employed
in the proof of Theorem \ref{theo:maintheo1} (one dimension) is different
from the one described above, as the same particle is ``moved'' several times
in the course of the procedure and, every time it fills an empty 
site, it paves the way to the particle that are moved subsequently. 
This allows to prove a stronger result, i.e., that
activity is sustained at arbitrarily low density
by setting $\lambda$ small enough. 

The proof of Theorem \ref{theo:depend} uses the stabilization
procedure that has been developed by Rolla and Sidoravicius \cite{Rolla}
and it is based on an observation. We refer the reader 
to Section \ref{sect:lowerbound}.

\section{Diaconis-Fulton representation}
\label{sect:Diaconis}
In this section we describe the
Diaconis-Fulton graphical representation
for the dynamics of ARW.
We follow~\cite{Rolla}. 
Let $\eta \in {\mathbb{N}_{0\rho}}^{\mathbb{Z}^d}$
denote the particle configuration,
where  $\mathbb{N}_{0\rho} = \mathbb{N}_0 \cup \{ \rho \}$.
We define an order relation for $\rho$, which 
represents the presence of an $S$-particle at one site,
setting
$0 < \rho < 1 < 2 \ldots$. 
We also let
$ |\rho| = 1$, so that
$|\eta_t(x)|$ counts the number of particles regardless of their state.
The addition is defined by $\rho + 0 = \rho$, 
and $\rho + k = k  + 1 $ if $k \geq 1$,
providing the $A + S \rightarrow 2A$ transition.
The $A \rightarrow S$ transition is represented by 
$\rho \cdot k$, where $\rho \cdot 1 = \rho$ 
and $\rho \cdot k = k$ if $k \geq 2$.
We introduce two operators, ``move'' from $x$ to $y$,
which is denoted by $\tau_{xy}$,
 and ``sleep'' at $x$, which is denoted by $\tau_{x\rho }$.
These operators act on the particle configuration.
%For each site $x$, we have the transitions 
%$\eta \rightarrow \tau_{xy} \eta$ at rate $[ \eta_t(x) ]^{*}\,  p ( y - x )$,
For any $\eta \in  \mathbb{N}_{0\rho}^{\mathbb{Z}^d}$,
the configuration $\tau_{xy} \eta \in \mathbb{N}_{0\rho}^{\mathbb{Z}^d}$
is defined as,
\begin{equation}
\tau_{xy} \eta (z) = 
\begin{cases} 
\eta(z) + 1  &\mbox{if }  z = y, \\ 
\eta(z) - 1 & \mbox{if } z=x, \\
\eta(z)  & \mbox{if } z \neq x  \mbox{ and } z \neq y, \\
\end{cases}
\end{equation}
and the 
configuration $\tau_{x\rho} \eta \in \mathbb{N}_{0\rho}^{\mathbb{Z}^d}$
is defined as,
\begin{equation}
\tau_{x\rho} \eta (z) = 
\begin{cases} 
\eta(z) \cdot \rho  &\mbox{if }  z = x, \\ 
\eta(z) & \mbox{if } z \neq x. \\
\end{cases}
\end{equation}
A site $x \in \mathbb{Z}^d$ is \textit{stable} in the configuration $\eta$ if
$\eta(x) \in \{0, \rho \}$
and it is \textit{unstable} if $\eta(x) \geq 1$.
We fix an array of  \textit{instructions} 
$\tau = ( \tau^{x,j}: \, x \in \mathbb{Z}^d, \, j \in \mathbb{N})$,
where $\tau^{x,j}= \tau_{xy}$ 
or $\tau^{x,j}= \tau_{x\rho}$.
Let $h = ( h(x)\, : \,  x \in \mathbb{Z}^d)$ count the number of 
instructions used at each site.
We say that we \textit{use} an instruction
at $x$ when we act on the current
particle configuration $\eta$ through the operator $\Phi_x$,
which is defined as,
\begin{equation}
\label{eq:Phioperator}
\Phi_x ( \eta, h) =
( \tau^{x, h(x) + 1}  \, \eta, \, h + \delta_x).
\end{equation}
The operation $\Phi_x$ is \textit{legal} for $\eta$ if $x$ is unstable in $\eta$, i.e.,
$\eta(x) \geq 1$, otherwise it is \textit{illegal}.

\vspace{0.8cm}

\noindent \textit{\textbf{Properties.}}
We now describe the properties of this representation.
Later  we discuss how they are related to the the stochastic dynamics of ARW.
For $\alpha = ( x_1, x_2, \ldots x_k)$,
we write $\Phi_{\alpha} = \Phi_{x_k} \Phi_{x_{k-1}}
\ldots \Phi_{x_1}$ and we say that $\Phi_{\alpha}$ is
\textit{legal} for $\eta$ if $\Phi_{x_l}$
is legal for $\Phi_{(x_{l-1}, \ldots, x_1)} (\eta,h) $
for all $l \in \{ 1, 2, \ldots k \}$.
Let $m_{\alpha} = ( m_{\alpha}(x) \, : \,x \in  \mathbb{Z}^d )$
be given by,
 $m_{\alpha}(x) \, = \, \sum_{l} \mathbbm{1}_{x_l = x},$
the number of times the site $x$ appears in $\alpha$.
We write $m_{\alpha} \geq m_{\beta}$ if
$m_{\alpha} (x)  \,  \geq \, m_{\beta} (x) \, \, \, \forall x \in \mathbb{Z}^d$.
Analogously we write $\eta^{\prime}   \geq   \eta$ if $\eta^{\prime} (x) \, \geq \, \eta(x)$
for all $x \in \mathbb{Z}^d$. We also write $(\eta^{\prime}, h^{\prime}) \geq (\eta, h)$
if $\eta^{\prime} \geq \eta$ and $h^{\prime} = h$.
Let $\eta, \eta^{\prime}$ be two configurations, $x$ be a site in $\mathbb{Z}^d$
and  $\tau$ be a realization
of the set of instructions. 
Let V be a finite subset of $\mathbb{Z}^d$. A configuration $\eta$ is said to be \textit{stable} in $V$
if all the sites $x \in V$ are stable. We say that $\alpha$ is contained in $V$
if all its elements are in $V$ and we say that $\alpha$ \textit{stabilizes} $\eta$ in $V$
if every $x \in V$ is stable in $\Phi_\alpha \eta$.
For the proof of the following Lemmas we refer to~\cite{Rolla}.

\paragraph*{Lemma 1}
\label{prop:lemma2}
(Abelian Property)
If $\alpha$ and $\beta$ are both legal sequences for $\eta$
that are contained in $V$ and stabilize $\eta$ in $V$, 
then $m_{\alpha} = m_{\beta}$. In particular, $\Phi_{\alpha} \eta = \Phi_{\beta} \eta$.

\vspace{0.5cm}

\noindent By Lemma $2$, $m_{V, \eta, \tau} = m_{\alpha}$ and $\xi_{V, \eta, \tau} = \Phi_{\alpha} \eta$ are well defined.
\paragraph*{Lemma 2}
\label{prop:lemma3}
(Monotonicity)
If $V \subset V^{\prime}$ and $\eta \leq \eta^{\prime}$, then $m_{V, \eta, \tau} \leq m_{V^{\prime}, \eta^{\prime}, \tau}$.

\vspace{0.5cm}

\noindent By monotonicity, the limit 
$$m_{\eta, \tau} = \lim\limits_{ V \uparrow \mathbb{Z}^d} m_{V, \eta, \tau},$$ exists and does not depend
on the particular sequence $V \uparrow \mathbb{Z}^d$.

We now introduce a probability measure on the space of instructions and of particle configurations.
We denote by $\mathcal{P}$ the probability measure according to which,
for any $x \in \mathbb{Z}^d$, $j \in \mathbb{N}$,
$\mathcal{P} (  \tau^{x,j} = \tau_{xy}   ) = \frac{p(y-x)}{1 + \lambda}$
and  $\mathcal{P} (  \tau^{x,j} = \tau_{x\rho}   ) = \frac{1}{1 + \lambda}$ independently.
Finally we denote by $\mathcal{P}^\nu$ the joint law of
$\eta$ and $\tau$, where 
$\eta$ has distribution $\nu$ and 
it is independent from 
$\tau$. The following lemma relates the dynamics of ARW to the stability property of the representation.
\paragraph*{Lemma 3}
\label{prop:lemma4}
Let $\nu$ be a translation-invariant, ergodic distribution with finite density $\nu(\eta( \mathbf{0}))$.
Then $\mathbb{P}^{\nu}  (\mbox{ ARW fixates locally } ) = \mathcal{P}^{\nu} ( m_{\eta, \tau} (\mathbf{0}) < \infty ) \in \{0, 1 \}$.

\vspace{0.5cm}

\noindent
The next lemma states that by replacing an instruction ``sleep'' by a neutral instruction the number of instructions
used at the origin for stabilization cannot decrease.
Thus, besides the $\tau_{xy}$ and $\tau_{x\rho}$,
consider in addition the neutral instruction $\mathcal{I}$, 
given by $\mathcal{I} \, \eta = \eta$.  
Given two arrays $\tau = \left( \tau^{x,j} \right)_{x ,\, j }$ 
and $\tilde{\tau} = \left( \tilde{\tau}^{x,j} \right)_{x, \, j }$,
we write $\tau \leq \tilde{\tau}$ if for every $x \in \mathbb{Z}^d$ and $j \in \mathbb{N}$,
either $\tilde{\tau}^{x,j} = {\tau}^{x,j}$ or $\tilde{\tau}^{x,j} = \mathcal{I}$ and 
${\tau}^{x,j} =  \tau_{x\rho}$.

\paragraph*{Lemma 4}
\label{prop:lemma5}
(Monotonicity with enforced activation)
Let $\tau$ and $\tilde{\tau}$ be two arrays of instructions such that $\tau \leq \tilde{\tau}$.
Then, for any finite $V \subset \mathbb{Z}^d$ and $\eta \in \mathbb{N}_{0\rho}^{\mathbb{Z}^d}$,  
$m_{V, \eta, \tau} \leq m_{V, \eta, \tilde{\tau}}.$

\section{Proof of Theorem \ref{theo:maintheo1}}
\label{sect:upperbound1}
Without loss of generality we assume $\mathbf{m} >0$ and we consider the set $B_L = [-2L,\, 0]$.
The case $\mathbf{m} < 0$ can be recovered by reflection symmetry.
We stabilize only particles in $[-L, 0]$, but we consider the site $-2L-1$ as the outer boundary of the set,
i.e., once a particle is on a site   $ \leq -2L-1$ it is ``lost''.

Let $\tilde{N}_0^{L}$ be the number of particles in $[-L,0]$.
First, we ``move'' every particle starting in $[-L,0]$ until every site of $[-L,0]$ is
either empty or it hosts only one active particle.
This means that if the site hosts initially $n>1$ particles, we move $n-1$ particles
until each of them fills an empty site.
By ``moving'', we mean that we always use the instruction on the site where the particle is located 
until the particle reaches an empty site.
Now, every site in $[-L,0]$ either hosts one particle or is empty.
Let $N_0^L$ be the number of particles in $[-L,0]$.
The next proposition states that with uniformly positive probability
we loose  a number of particles that is bounded from above by a number that
not depend on $L$.
\begin{proposition}
\label{prop:initial}
There exist two positive constants $c$ and $K$ such that for all $L \in \mathbb{N}$,
\begin{equation}
\label{eq:prop}
\, \, 
\mathcal{P}^{\nu} \,   ( \, \tilde{N}_0^{L}- N_0^L \,  \leq  c  \, \,  ) 
\,  \geq K.
\end{equation}
\end{proposition}
\begin{proof}[Proof of Proposition \ref{prop:initial}]
Since we are only moving particles that are not alone,
this is equivalent to the model with $\lambda=\infty$.
By~\cite{Cabezas}[Theorem 4], at $\lambda=\infty$,
there is fixation for any $\mu < 1$. Therefore,
$m_{[-L,0], \eta, \tau}(z)$ is a finite random variable for any 
$z \in [-L,0]$, and thus the sum $\sum_{z}m_{[-L,0], \eta, \tau}(z)$
for $z$ on the inner boundary of an interval
$[-L,0]$ is tight with respect to $L$.
Since each particle leaving $[-L, 0]$
must perform a jump from a site of its inner boundary,
the result follows.
\end{proof}

\noindent Now every site in $[-L,0]$ hosts at most one particle, which is necessarily active. 
We stabilize the set $[-L, 0]$ according to the following rule.
Let $z_0=-L$. If the site is empty, we do not do anything.
If $z_0$ hosts one particle, then we move it until one of the following events occurs:
\textbf{(1)} the particle sleeps somewhere in $[-2L, z_0]$, 
\textbf{(2)} the particle reaches a site  $x \leq -2L-1$, 
\textbf{(3)} the particle reaches the first empty site in $[z_0+1, 0]$,
\textbf{(4)} the particle reaches a site $x \geq 0$.
If $(3)$ or $(4)$ occur, we say that a \textit{successful jump} has been performed.

As the random walk is biased to the right, we can uniformly bound from below by a constant $F_L$ the probability of a successful jump.
Indeed, consider now a random walk $\left( Z(j) \right)_{j \in \mathbb{N}}$ starting from $Z(0) = z_0$ in the following environment. Namely, if $y > z_0$ then the walker located at $y$ jumps to $y+z$ with probability $p(z)$. If $y \leq z_0$, then the walker jumps to $y + z$ with probability $\frac{p(z)}{1 + \lambda}$ and it sleeps with probability $\frac{\lambda}{1 + \lambda}$. As the random walk $(Z(j))_{j \in \mathbb{N}}$ can sleep on \textit{any} site in $(z_0 -L, z_0]$ and as 
$z_0 - L \geq -2L$, then the probability of a successful jump in the activated random walk model cannot be smaller than $F_L$.

Now let $z_1 = z_0 + 1$ and observe that every site in $[z_1, 0]$ is either empty or it hosts one active particle.
Let $N^L_1$ be the number of particles in $[z_1,0]$.
If $z_1$ hosts no particles, we do not do anything. 
Instead, if $z_1$ hosts one particle, we move such a particle as before, until one of the four events above occurs.
Again, a successful jump occurs with probability at least $F_L$.
We then define $z_2 =z_1 + 1$ and we continue in this way until we reach $z_L$.
We observe that, at every step $i$, 
$N^L_{i+1} = N^L_i$ with probability at least $F_L$ and $N^L_{i+1} = N^L_i-1$ with probability at most
$1-F_L$. 

Now we define  $F:= \lim\limits_{L \rightarrow \infty}  F_L,$ which corresponds
to the constant  (\ref{eq:defF}) defined before the statement of the theorem.
We observe that for any positive real $\epsilon$,
$N_0^L \geq ( \mu  - \epsilon ) L $ 
and $N^L_L \geq  N_0^L -  (1-F + \epsilon)   L  = ( \mu - 1 + F - 2 \, \epsilon)  L$ 
with high probability as $L$ is large enough.
Thus, for any positive $\delta$ such that $\mu = 1 - F + \delta$,
we let $\epsilon := \frac{\delta}{3}$ and we conclude that
$N_L^L \geq \frac{\delta}{3}  L$ with high probability.
Now, observe that $N^L_L$ corresponds to the number of particles that left the set $[-2L,0]$ from the right boundary.
In case of jumps on nearest neighbours, each of these particles must have crossed the origin. 
In case of biased distribution with general (finite) support, the same conclusion does not hold. 
Thus, let $Q_L := \{ z \in  [-L,0] \, : \, \exists x \in  \mathbb{Z} \setminus [-L,0] \, \mbox{ s.t. } p(x-z)>0\}$
be the inner boundary of $B_L$ and let $K_2$ be a constant such that $|Q_L| \leq K_2$ for every $L$.
Thus, as at least $N^L_L$ particles left the set $[-2L,0]$, then
$\exists z  \in Q_L$ such that $m_{[-2L,0], \eta, \tau} (z) \geq \frac{\delta}{3 K_2}  L$ with high probability. 
By the union bound, this implies that there exists a site $z \in Q_L$ such that
for every $L$ large enough,
\begin{equation}
\mathcal{P}^{\nu} \left( m_{[-2L,0], \eta, \tau} (z)  \geq \frac{\delta}{3 K_2}  L \right) \geq \frac{1}{2K_2}.
\end{equation} 
Thus, by using translation invariance and by Lemma 3 we conclude that ARW stays active almost surely.
\qed

\section{Proof of Theorem \ref{theo:maintheo2}}
\label{sect:upperbound2}
We present the proof in the case of two dimensions.
The same arguments can be adapted to the case of more than two dimensions.
We assume that $\mathbf{m} \cdot \mathbf{e}_1 >0$ and 
we introduce the set $B_L = \{  (x,y) \in \mathbb{Z}^2 \, : \, x \in [-L+1, 0], y \in [-L^3, L^3]  \}$.
We order the sites of $B_L$ by writing $B_L = \{z_1, z_2, \ldots z_{|B_L|}\}$,
requiring that sites with smaller $x$ appear first.
We stabilize the set $B_{2L}$, but we ``move'' only particles
which start from sites in $B_L$, as we want them to be ``far''
from the boundary of the set.
By ``moving'', we mean that 
we always use the instruction
on the site where the particle is located until 
a certain event occurs.
In our stabilization procedure, we say that a particle
is ``good'' if it occupies one of the sites that is empty
for the initial configuration or if it leaves $B_L$ by crossing
the line $x=0$. Because of the bias and of the order
according to which particles are moved,
we can provide a positive uniform lower bound for the probability of a particle being good.
The general goal of the proof is to show that, if the density of empty sites 
for the initial configuration is less than the density of good particles,
then a positive density of particles must leave $B_L$ by crossing the line $x=0$.
We use translation invariance then to show that at least $ c  L$ particles
cross the origin with high probability for some $c>0$,
which in turn implies almost sure activity by Lemma 3.

The stabilization procedure is defined as follows.
We consider the first site in the order, $z_1 = (x_1, y_1)$, and we move one of its 
particles until one of the following events occurs. Namely,
\begin{enumerate}
\item[(1)] either the particles reaches one empty site $(x,y)$ such that $x > x_1$
\vspace{-0.15cm}
\item[(2)] either the particle leaves $B_L$,
\vspace{-0.15cm}
\item[(3)] or the particles uses an instruction ``sleep'' on a site $(x,y)$ such that $ x \leq x_1$.
\end{enumerate}
Then, we consider the other particles on the same
site and for each of them we employ the same procedure.
At the next step,  we consider the second site $z_2$ in the order we repeat
the same procedure for all its particles.
We proceed in this way until all the particles have been moved one time.

We let $\mathcal{N}_L$ be the number of particles that visit the origin at least one time.
Clearly, $m_{B_L, \eta, \tau}(0) \geq \mathcal{N}_L$.
In order to estimate $\mathcal{N}_L$, we adapt 
the idea of ``ghost'' explorers 
\cite{Lawler, Shellef} to our setting.
Namely, every time a particle starting from $z_i =(x_i, y_i)$ stops at an empty site
$(x,y)$ (which, by definition of stabilization procedure, must satisfy $x > x_i$),
we let a ghost start from $(x,y)$ and perform a random walk until it reaches
the inner boundary of $B_{2L}$, i.e., $\partial^{i} B_{2L} := \{ x \in B_{2L} \, \,\mbox{ s.t. }\, \, \exists \,\, y \in \mathbb{Z}^2 \setminus B_{2L}
\mbox{ and } y \sim x\}$.  Ghosts do not interact with other particles.
We let $W_L$ be the number of particles visiting the origin as a ghost or as an original particle and we let
$R_L$ be the number of particles visiting the origin only as a ghost. Then,
\begin{equation}
\mathcal{N}_L \,{\buildrel d \over =}\, W_L - R_L.
\end{equation}
The variables $W_L$ and $R_L$ are of course dependent. We first provide sufficient
conditions for $E[W_L] - E[R_L] \geq c  L$ for some $c>0$
and we then prove that such a condition implies
that $\mathcal{N}_L \geq \frac{c}{3}  L$ with high probability.

We now provide an estimation of the expectations of $W_L$ and $R_L$.
For any $z \in B_{2L}$ and for any $j \in \mathbb{N}$, we 
introduce the sequence
$\{ S^{z,j}(t), Y^{z,j}(t) \}_{t \in \mathbb{N}}$,
where $S^{z,j}(t)$ is a random walk with jump distribution
$p(\, \cdot \, )$ and starting from
$z$  and $\{\, Y^{z,j}(t) \, \}_{t \in \mathbb{N}} $
 is an infinite sequence
of independent and identically distributed
random variables such that $Y^{z,j}(0) = 1$ with probability
$\frac{\lambda}{1+ \lambda}$ and 
$Y^{z,j}(0) = 0$ with probability $\frac{1}{1 + \lambda}$.
We start with the estimation of $E[ W_L ]$.
Thus, we let from every particle
$(z,j)$, $z=(x,y) \in B_L$, $1 \leq j \leq \eta(z)$,
a simple random walk start and we count the number of them
visiting the origin before leaving $B_{2L}$ and before using
any instruction sleep on the set
$H_{x} := \{(x^{\prime}, y^{\prime}) \in \mathbb{Z}^2 \, : \, x^{\prime} \leq x \}$, i.e.,
\begin{align}
\label{eq:WL}
W_L \,{\buildrel d \over \geq}\,  \tilde{W}_L & := \sum\limits_{z \in B_{L}}  \sum\limits_{ 1 \leq j \leq \eta(z)} \mathbbm{1} ( \{ S^{z,j}( \tau_{\partial^{i} B_{2L} } ) = \mathbf{0}   \}  
\\ & ~~~~~~ 
\cap \{ \nexists t \leq \tau^{z,j}_{ \partial^{i} B_{2L}} \, \mbox{ s.t. } \, Y^{z,j}(t) = 1 \mbox{ and } S^{z,j}(t) \in H_x    \}    )
\end{align}
where $\mathbbm{1}( \, \cdot \, )$ is the indicator function, $\eta$ is the initial particle configuration
and $z = (x,y)$, $\tau^{z,j}_{ \{ \, \cdot \,\} }$ is the hitting time of $\{ \cdot \}$ for the random walk $X^{z,j}$.
The (stochastic) inequality holds as on the right-hand side 
we count only the walks that hit the inner boundary of $B_{2L}$ for the first time
at the origin and as, once the particle starting from $(x,y)$ turns to a ghost somewhere, 
it can explore the region $H_x$ without any restriction
related to the outcome of the instructions sleep. Thus, the condition
on the right-hand side is more restrictive.

The term $R_L$ is more difficult to handle. 
However, note that every ghost necessarily
starts its walk from a site of $B_L$ that is empty in the initial configuration $\eta$,
due to the order according to which particles are moved.
Thus, we provide a (stochastic) upper bound for $R_L$
by letting for every empty site  a random walk start and by counting
the number of them  hitting the inner boundary of $B_L$ at the origin, 
without any further restriction. We denote such a number 
by $\tilde{R}_L$. Therefore,
\begin{equation}
\label{eq:RL}
{R}_L \,{\buildrel d \over \leq}\, \tilde{R}_L = \sum\limits_{z \in B_{L}}  \mathbbm{1} \left( X^{z,j}( \tau_{\partial^{i} B_{2L} } ) = \mathbf{0}   \right)  \mathbbm{1} \left( \eta(z)=0  \right) 
\end{equation}

We let now $G_K = \{(x,y) \in \mathbb{Z}^2 \, \mbox{ s.t. } x=k\}$ and
$D_k = \{(x,y) \in \mathbb{Z}^2 \, \mbox{ s.t. } y=k\}$.
By using independence and translation invariance,
\begin{align*}
E[ \tilde{W}_L ] 
& =  \mu  \sum\limits_{(x,y) \in B_{L}} P  (  \{ S^{(x,y)}( \tau_{\partial^{i} B_{2L} } ) = \mathbf{0}   \} 
\\ & ~~~~~~~~~~~~~~~~ \cap \{ \nexists t \leq \tau^{(x,y)}_{ \partial^{i} B_{2L}} \, \mbox{ s.t. } \, Y^{(x,y)}(t) = 1 \mbox{ and } S^{(x,y)}(t) \in H_x    \}    )
\\   & \geq \mu  \sum\limits_{x = -L+1} \sum\limits_{y= -L^2}^{L^2} P ( \{ S^{(x,y)} \mbox{ reaches $G_0$ at $\mathbf{0}$ before
reaching $y + D_{L^2}$, $y + D_{-L^2}$ and $x  + G_{-L}$  }  \}  
\\ & ~~~~~~~~~~~~~~~~ \cap \{ \nexists t \in \mathbb{N}\, \mbox{ s.t. } \, Y^{(x,y)}(t) = 1 \mbox{ and } S^{(x,y)}(t) \in H_x    \}    ) 
\\   & = \mu  \sum\limits_{x= -L+1} 
\sum\limits_{y= -L^2}^{L^2} P ( \{ S \mbox{ reaches $G_{-x}$ at (-x,- y) before
reaching $D_{L^2}$, $D_{-L^2}$ and $ G_{-L}$  }  \}  
\\ & ~~~~~~~~~~~~~~~~ \cap \{ \nexists t \in \mathbb{N}\, \mbox{ s.t. } \, Y(t) = 1 \mbox{ and } S(t) \in H_0   \}    ) 
\\ &  \geq L  \mu  P ( \{ S \mbox{ reaches $G_{L}$ before
reaching $D_{L^2}$, $D_{-L^2}$,  and $G_{-L}$  }  \}  
\vspace{2cm}
\\ & ~~~~~~~~~~~~~~~~ \cap \{ \nexists t \in \mathbb{N}\, \mbox{ s.t. } \, Y(t) = 1 \mbox{ and } S(t) \in H_0   \}    ) .
\end{align*}
Note that we omitted any superscript for the random walk starting from the origin.
Observe that the last inequality holds as the sum is over the probability of disjoint events and as the condition
on the right-hand side is more restrictive.
By the law of large numbers and as the random walk spends only a finite amount
of time in $H_0$, the probability of the event in the right-hand side of the last inequality converges to $F(\lambda, p(\, \cdot \, ))$
as $L \rightarrow \infty$, which is defined before the statement of the theorem.
By using the same arguments, we obtain the corresponding equation for
$E[ \tilde{R}_L ]$,
\begin{align*}
\label{eq:ERL}
E[ \tilde{R}_L ] 
& =  (1-\mu)   \sum\limits_{(x,y) \in B_{L}} P  (  \{ S^{(x,y)}( \tau^{(x,y)}_{\partial^{i} B_{2L} } )  = \mathbf{0}\} ) \\
& \leq  (1-\mu)   \sum\limits_{x=-L+1}^{0} \sum\limits_{y=-L^3}^{L^3}  P  (  \{ S^{(x,y)}  \mbox{ hits $G_0$ at the origin }  \} ) \\
& =  (1-\mu)  \sum\limits_{x=-L+1}^{0} \sum\limits_{y=-L^3}^{L^3}  P  (  \{ S \mbox{ hits $G_{-x}$ at $(-x,-y)$ }  \} )  \\
& \leq  (1-\mu)  \sum\limits_{x=-L+1}^{0} \sum\limits_{y=-\infty}^{\infty}  P  (  \{ S \mbox{ hits $G_{-x}$ at $(-x,-y)$ }  \} \} )  \\
& \leq  (1-\mu)   L.
\end{align*}
Thus, if $\mu > \frac{1}{1+F(\lambda)}$, then for all $L$ large enough,
$E[ \tilde{W}_L ] - E[ \tilde{R}_L ] \geq \frac{[ \mu \,  F(\lambda, p(\, \cdot \, )) - (1-\mu) ]}{2}  L $.
By using the union bound, the Chebyshev inequality and by observing tha the variance of $\tilde{W}_L$
and $\tilde{R}_L$ can be bounded by their expectation, we prove that
$\mathcal{N}_L \geq  \frac{[ \mu \,  F(\lambda, p(\, \cdot \, )) - (1-\mu) ]}{6}  L$ with high probability, which
in turn implies that at ARW stays active almost surely by Lemma 3.
Indeed, let $c = \frac{[ \mu \,  F(\lambda) - (1-\mu) ]}{2}$,
\begin{equation}
\label{eq:unionbound}
\begin{split}
& {P}
 (W_L - R_L < \frac{c}{3}  L  ) \,  \leq  \, {P} (\tilde{W}_L - \tilde{R}_L < \frac{E [\tilde{W}_L - \tilde{R}_L]}{3}  ) \,    \\
\leq & {P} ( \tilde{W}_L - E [\tilde{W}_L] >\frac{E [\tilde{W}_L - \tilde{R}_L]}{3}) 
+ {P} ( \tilde{R}_L - \mathbb{E} [\tilde{R}_L] > \frac{E [\tilde{W}_L - \tilde{R}_L]}{3}) \\
\leq & {P} ( \tilde{W}_L - E [\tilde{W}_L] >\frac{E [\tilde{W}_L - \tilde{R}_L]}{3}) 
+ {P} ( \tilde{R}_L - \mathbb{E} [\tilde{R}_L] > \frac{E [\tilde{W}_L - \tilde{R}_L]}{3})
\end{split}
\end{equation}
For the second inequality we used the union bound.
We now use the Chebyshev inequality and 
the inequalities $Var[\tilde{W}_L] \leq  E[\tilde{W}_L]$ and $Var[\tilde{R}_L] \leq E[\tilde{R}_L]$,
which hold as $\tilde{W}_L$ and $\tilde{R}_L$ are the sum of random variables taking values $0$ or $1$.
Thus, from (\ref{eq:unionbound}),
\begin{equation}
\label{eq:bound2}
\begin{split}
 {P}
 (W_L - R_L < \frac{c}{3} L  ) \, & \leq  
 9 \frac{Var[\tilde{W}_L]}{E[\tilde{W}_L-\tilde{R}_L]^2} + 9  \frac{Var[\tilde{R_L}]}{E[\tilde{W}_L-\tilde{R}_L]^2} \\ 
 & \leq 9 \frac{E[\tilde{W}_L]}{E[\tilde{W}_L-\tilde{R}_L]^2} + 9  \frac{E[\tilde{R}_L]}{E[\tilde{W}_L-\tilde{R}_L]^2}  \\
  & \leq  \frac{18}{c^2   L}.
  \end{split}
\end{equation}
and, by taking the limit $L \rightarrow \infty$, this concludes the proof of the theorem.
\qed

\section{Lower bound}
\label{sect:lowerbound}
\begin{proof}[Proof of Theorem \ref{theo:depend}]
We provide a new lower bound for
 $\mu_c(\lambda, q)$ 
 and we show that $\mu_c(\lambda,q) > \frac{\lambda}{1+\lambda}$ if 
 $q \not\in \{0, 1\}$.
This implies the statement of the theorem, as 
from \cite{Hoffman}
it is known that $\mu_c(1, \lambda)  = \mu_c(0, \lambda) = \frac{\lambda}{1+\lambda}$.

Our goal is to estimate under which conditions on $\mu$, $\lambda$ and $q$ the next
condition holds,
\begin{equation}
\label{eq:lowerbound}
\exists \, c > 0 \mbox{ s.t.}\, \forall L \in \mathbb{N}, \, \, \, 
\mathcal{P}^{\nu}( m_{{V_L}, \eta, \tau} (0)= 0  ) > c,
\end{equation}
where $V_L = [-L,L]$.
Indeed, from Lemma 3,  (\ref{eq:lowerbound})
implies that ARW  fixates almost surely.
In order to prove \ref{eq:lowerbound}, we consider
the stabilization of $[-L,0]$ and of $[0,L]$ separately.
Indeed, observe that, by independence of instructions,
\begin{multline}
%\begin{split}
\label{eq:product}
\mathcal{P}^{\nu}( m_{ [-L, L], \eta, \, \tau} (0)= 0  )\geq      \mathcal{P}^{\nu}( m_{ [-L, -1],  \eta, \tau} (0)= 0  )
\, \mathcal{P}^{\nu}( m_{[1, L], \eta,\, \tau} (0)= 0  ) \,  \nu(\eta(0)=0),
\end{multline}
as for any instruction array $\tau$ and $\eta \in \Sigma$,
$$ m_{[-L, -1], \eta, \tau } (-1)= 0,  \, m_{[1, L], \eta, \tau} (1)= 0, 
  \mbox{ and } \eta(0) = 0   \implies  m_{ [-L, L], \eta, \tau} (0) = 0.$$

Without loss of generality, we consider $q \leq 1/2$.
Indeed, the case of $q\geq1/2$ can be recovered by reflection symmetry.
First, we consider the stabilization of $[-L, -1]$.
If $q < \frac{1}{2}$ and $V_L=[-L, -1]$, it is easy to prove that,
for any value of $\mu$ and $\lambda$,
(\ref{eq:lowerbound}) holds.
Indeed, recall that, by Lemma 4, by erasing from the instruction array all the instructions ``sleep''
on sites $x \leq 0$, the number of instructions used at the origin for stabilization can only increase.
Then, we move the particles in $x \leq 0$ one by one, until each of them leaves the set $[-L, -1]$.
The trajectory of each of them follows a simple random walk without any interaction, 
as the instructions ``sleep'' have been erased.
As the bias is to the left, the probability that no particle hits the origin is uniformly positive in $L$.

It remains to prove that (\ref{eq:lowerbound})
holds with $V_L = [1, L]$ and $q \leq \frac{1}{2}$.
For this, we modify the stabilization procedure that has been developed
by Rolla and Sidoravicius \cite{Rolla}, which is sketched in Section \ref{sub:RollaSid}.
Our stabilization algorithm is presented in Section  \ref{sub:ouralgorithm}.

\subsection{The stabilization procedure of Rolla and Sidoravicius}
\label{sub:RollaSid}
In this section we briefly describe the stabilization procedure
that has been developed by Rolla and Sidoravicius \cite{Rolla}.
The procedure explores a certain set of instructions of $\tau$
and identifies a suitable \textit{trap} for every particle.
The trap is a site where the particle finds an instruction ``sleep''
and turns to the $S$-state.
The trap is chosen in such a way that, when a particle is moved to its trap,
it does not wake up any of the particles that have already
turned to the $S$-state.
In the absence of a suitable trap, the algorithm fails. 
If a suitable trap is found for every particle, then  
we say that the algorithm is successful
and this implies that $m_{[0,L], \eta, \tau}(0) =0$.
The goal is to prove that the probability of success
is uniformly positive in $L$.

We let $X^1 \leq X^2 \leq \ldots \leq X^{N_L}$ be the
position of the particles in $[0,L]$ at time $0$,
ordered from the left to the right,
where $N_L$ is the total number of particles in $[0,L]$.
We assume $X^1 >0$, which occurs with positive probability.
We start from the leftmost particle in the set and we
``explore'' its putative trajectory until the origin is reached.
As the exploration starts from a site which is on the right of the origin,
the last ``explored'' instruction at any site 
must be ``go left''.
The trap is defined as the  \textit{leftmost instruction ``sleep'' among 
those right below the last instructions ``go left''}.  
We denote the site where the trap is located as $T^1$.
Then, the particle is moved until such an instruction ``sleep''
is reached. For this, all the instruction ``sleep'' 
belonging to the set of explored instructions 
and which are not the trap are ignored.
Lemma 4 guarantees that, if instructions
``sleep'' of $\tau$ are ignored, then the total
number of instructions that must be used at $0$
to stabilize $[0, L]$ cannot be smaller than 
$m_{[0,L], \eta, \tau}(0)$. 
This is important, 
as we need to provide sufficient conditions for $m_{[0,L], \eta, \tau}(0) = 0$.

At the second step, we consider the second leftmost particle in $[0,L]$.
Starting from $X^2$, we explore its putative trajectory
until the site $T^1$ is reached.
As before, we let the trap be the 
leftmost instruction ``sleep'' among 
those right below the last instructions ``go left''.
We let $T^2$ be the site where the trap of the second particle is located.
We move such a particle to its trap ignoring
all the instructions sleep on the way to the trap.

Moving from the left to the right,
we repeat this procedure for every particle in $[0, L]$. 
The algorithm fails when no suitable trap
is found for one particle.
This might occur only in two cases.
Namely, when we explore the putative trajectory of the particle
starting from $X^i$, 
if  no instruction ``sleep'' is found right below the last instruction
``go left'' at any of the explored sites 
or if such instruction ``sleep'' is found,
but it is not located on the left of $X^{i+1}$,
then the algorithm fails.

Note that not all the instructions belonging to the explored
path are ``used'' by the particle. 
Successful algorithm means that no particle
ever visits sites hosting instructions that belong to previous explorations
and that have not been used (\textit{corrupted region}).
Indeed, for all $i$, the region of explored sites for $X^i$ 
is always on the right of the trap $T^{i-1}$,
while the corrupted region
is on sites $\leq T^{i-1}$.
This is necessary to  have a control 
on the joint distribution 
of the outcome of different explorations
by using independence of instructions.

\subsection{Our algorithm}
\label{sub:ouralgorithm}
The difference between our stabilization algorithm and the one developed by Rolla
and Sidoravicius involves the criterion according to which the trap is chosen.
By looking only at the instructions located right below the last instruction ``go left'',
as in the algorithm by Rolla and Sidoravicius, 
one  ignores most of the instructions ``sleep'' which belong
to the set of explored instructions. 
In order to save space, we provide a different definition
of traps by taking into account for such instructions ``sleep'' as well.
This allows to stabilize particles closer one to the other than in~\cite{Rolla}.

We move from the leftmost particle in $[0, L]$ to the right
and we explore the putative trajectory of every particle, as before.
Our traps are defined as the 
\textit{last instruction ``sleep'' that has been discovered 
during the whole exploration} (without
requiring for it to be right below the last instruction ``go left'').
In order to separate the region of corrupted sites
from the region of unexplored sites, we introduce \textit{barriers}.
The barrier is defined as the \textit{rightmost site 
on the explored path that has been visited after the last instruction sleep}
(see Figure \ref{fig:exploration1} and \ref{fig:exploration2}).
We let  $T^i$ and $A^i$ be the 
site where the trap and the barrier 
of the $i$-th exploration are located respectively.
Every exploration is carried on until the barrier that has been
identified at the previous step is reached.
The barrier $A^i$ must always be on the left of 
$X^{i+1}$. If during the exploration no instruction
``sleep'' is found or if such an instruction
is found, but $A^i \geq X^{i+1}$,
then we declare the algorithm to have failed.
Thus, the barrier separates the corrupted region from
the space that is available for the next exploration.

Our stabilization procedure is sensitive to the bias 
of the jump distribution as, the weaker is the bias, 
the larger is the number of times the exploration visits the same site. 
This in turn implies that, the weaker is the bias, the 
higher is the chance of finding instructions ``sleep'' close to the previous barrier.
\begin{figure}[!]
\includegraphics[scale=0.40]{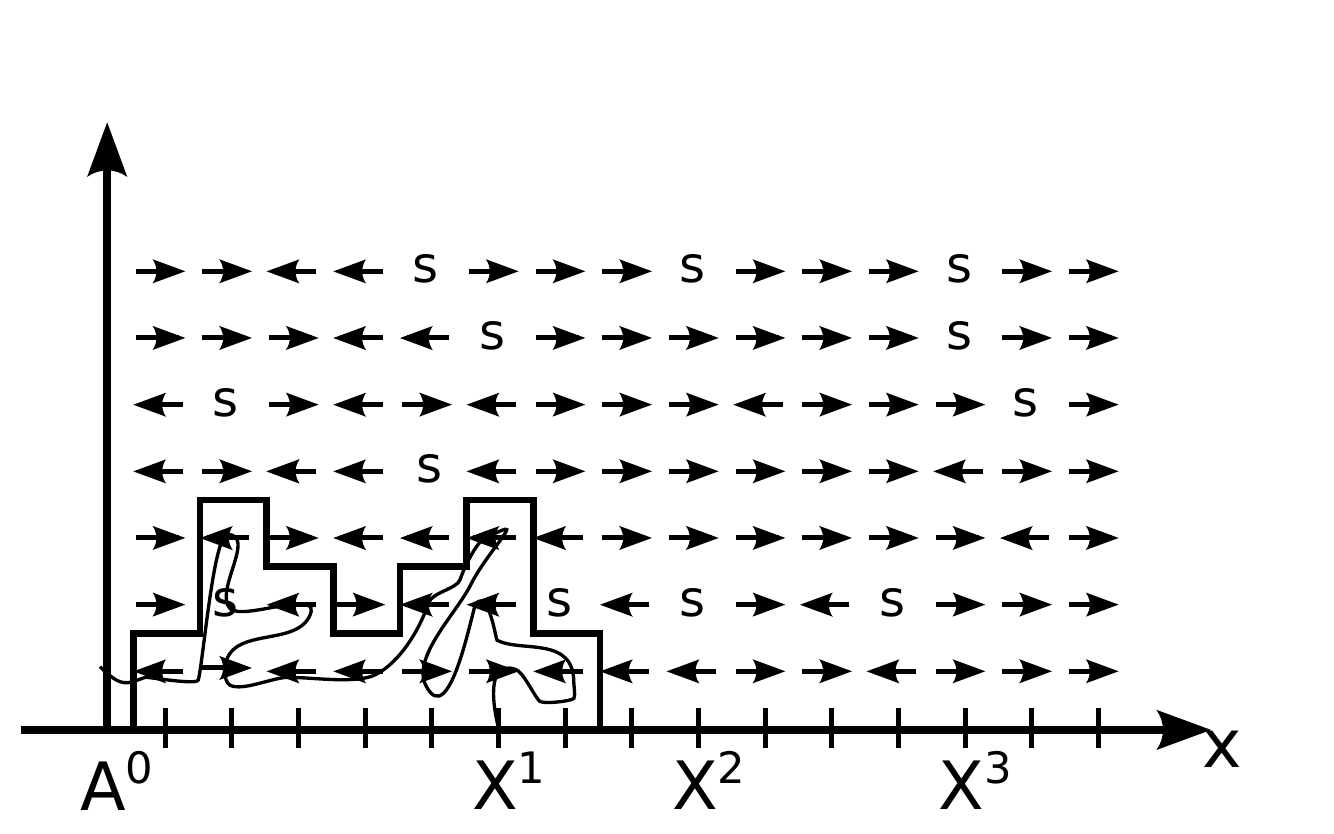}
\includegraphics[scale=0.35]{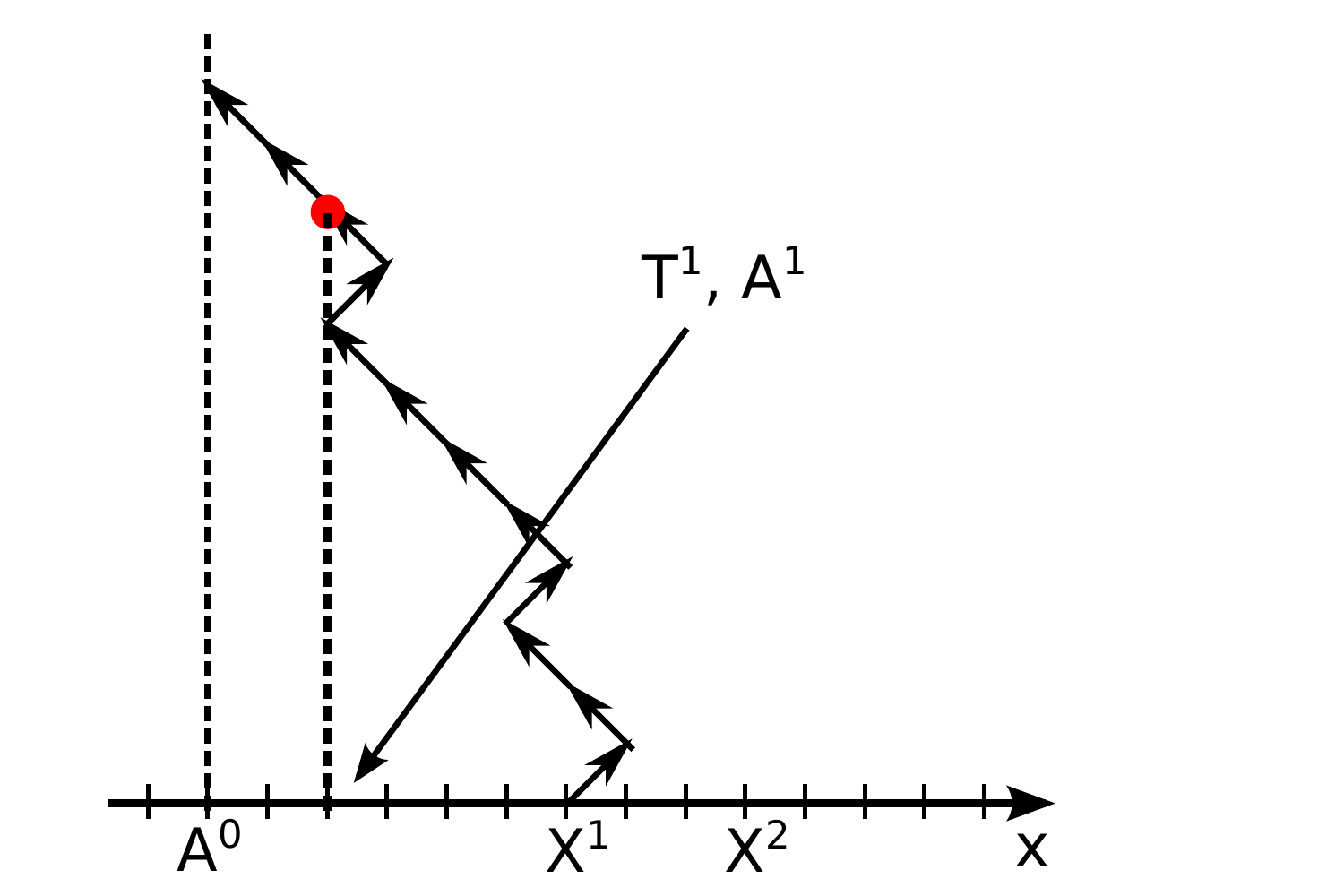}
\caption{Representation of the first exploration. \textit{Left:} instructions belonging to the first exploration. \textit{Right:}  representation of the first exploration as a simple random walk path. Red circles represent the steps of such a path that are related to the presence of an instruction ``sleep''. In the example in the figure, the trap and the barrier are identified with the same site. }
\label{fig:exploration1}
\end{figure}

\begin{figure}[!]
\includegraphics[scale=0.35]{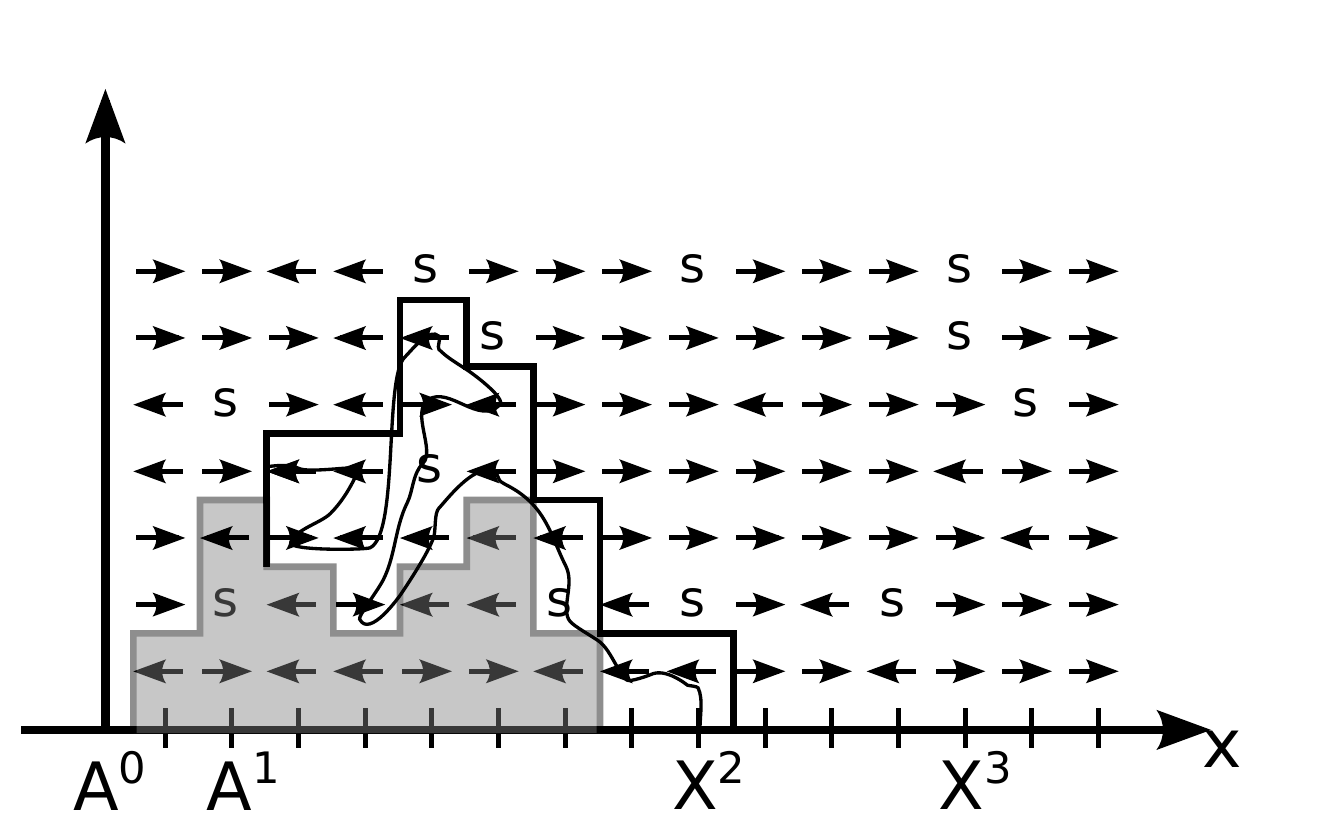}
\includegraphics[scale=0.35]{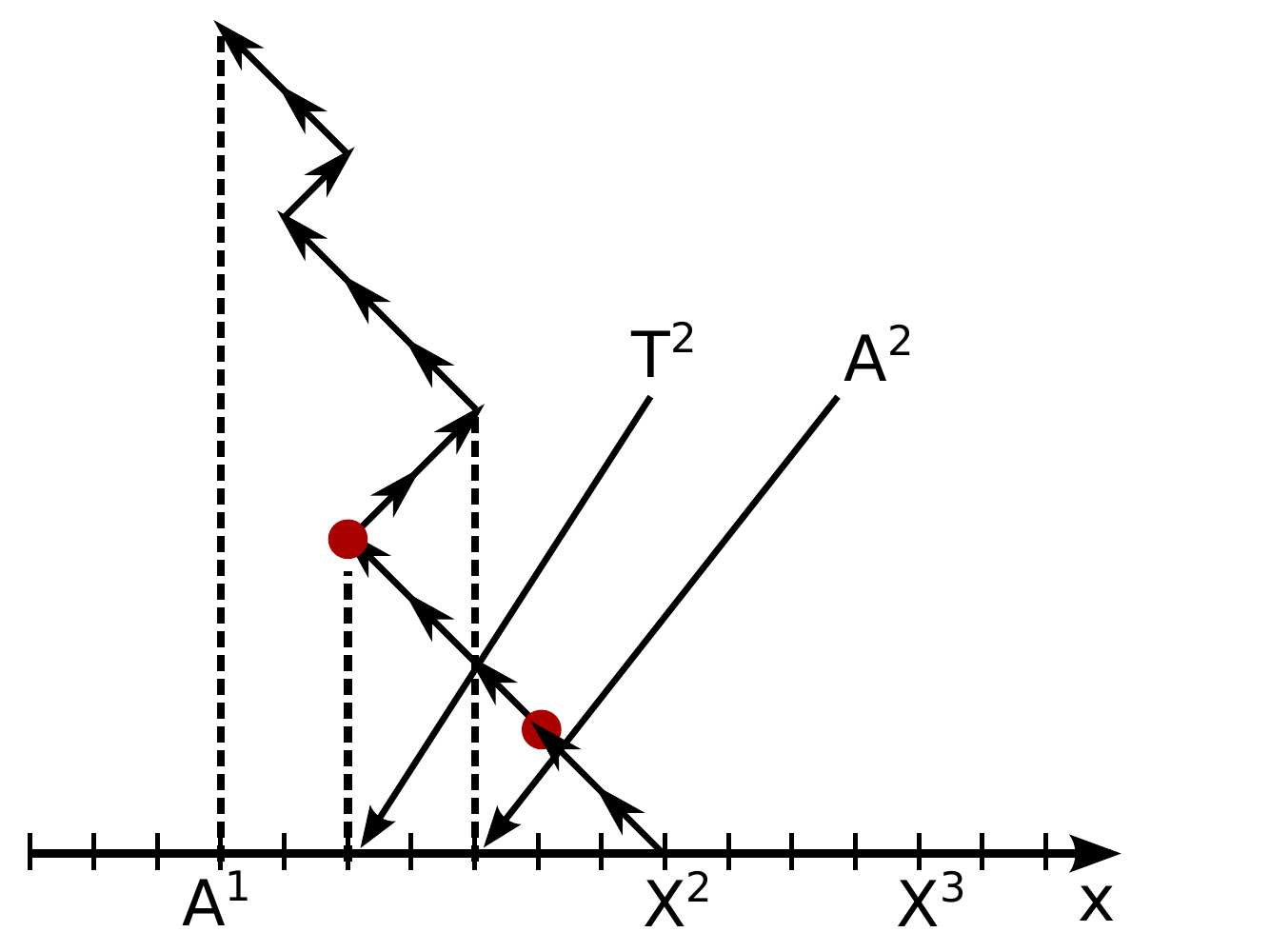}
\caption{Representation of the second step of the stabilization procedure.
\textit{Left:} the dark region represents the first exploration. The instructions below the continuous line in the non-dark region
represent the second exploration. \textit{Right:} representation of the second exploration as a simple random walk path. Red circles represent the steps of the path that are related to the presence of an instruction ``sleep''. Referring to the path in the figure as an example, according to the criterion employed in~\cite{Rolla} the trap would be taken as the site hosting the rightmost  instruction ``sleep'' between the two. Instead in our algorithm the trap is identified as the site denoted by $T^2$ in the figure. Furthermore, the barrier is identified as the site denoted by $A^{2}$. }
\label{fig:exploration2}
\end{figure}

\paragraph*{Probability of successful stabilization:} 
We let $X^1 \leq X^2 \leq \ldots \leq X^{N_L}$ be the positions of the particles at time $0$, ordered from the left to the right.
We let $A^i$ and $T^i$ be  the position of the barrier and of the trap for the particle $X^i$ respectively.

As success of the algorithm is a sufficient condition for $m_{[0,L], \eta, \tau}  (0) = 0$, then 
\begin{equation}
\label{eq:condition1}
\begin{split}
\mathcal{P}^{\nu}\left( m_{[0,L], \eta, \tau}  (0) = 0 \right) & \geq 
\mathcal{P}^{\nu} \left( 1 \leq \forall i \leq N_L, \, \, \,   A^{i} \leq X^{i} \right).
\end{split}
\end{equation}
We now prove that if $\mu< B(\lambda, q)$, where $B(\lambda, q)$ is a function such that
for every $\lambda$, $q \in \{0,1\}$, $B(\lambda, q) > \frac{\lambda}{1+\lambda}$,
then the right-hand site of (\ref{eq:condition1}) is uniformly positive in $L$.

The probability of success of the algorithm cannot increase with $L$,
as particles are ``killed'' at the boundary. 
Thus, for a lower bound for (\ref{eq:condition1}), 
we refer to the stabilization of the set $[0, \infty)$.
We claim that the position  $A^1$ of the first barrier follows a distribution
having expectation $E[A^1]$  which is such that
$E[A^1]  < \frac{1 + \lambda}{\lambda}$ if $q \not\in \{0, 1 \}$.
To be more precise, the same as in~\cite{Rolla}, the claim is that the probability
space can be enlarged so that we can define a random variable $Y^1$ independent of $\eta$
whose expectation $E[Y^1] $ has the property above
and such that the first step of the construction is successful only if $Y^1 \leq X^1$,
in which case the position $A^1$ of the first barrier is given by $A^1 = Y^1$.
Indeed, if at least an instruction sleep has been found in $[0, X^1]$ 
before hitting the barrier $A^0=0$,
we take $Y^1$ as the \textit{rightmost site
that has been visited starting from the last instruction sleep that has been found before
hitting $A^0$}.
Namely, we let $S^y(t)$ be a random walk starting from $y \in \mathbb{N}$ and we let
$\{ \, R(t) \, \}_{t \in \mathbb{N}}$ be a sequence of i.i.d. random variables
such that $R(0)=1$ with probability
$\frac{\lambda}{1 + \lambda}$ and $R(0)=0$ with probability $\frac{1}{1 + \lambda}$.
As after any exploration step the probability to ``discover'' an instruction
``sleep'' is $\frac{\lambda}{1+\lambda}$ independently, from the considerations above we conclude that,
for any $k \in \mathbb{N}$,
\begin{align*}
\mathcal{P}^{\nu} \left(Y^1=k \, | \, Y^1 \leq X^1 \right) 
&= \mathcal{P}^{\nu} \lr{  \max\{   x \in \mathbb{N} \, \mbox{ s.t. }  \, 
S^{X^1}(t)=x  \mbox{ for some } t \mbox{ s.t. } \tilde{\tau}^{X^1} \leq  t < \tau_0^{X^1}  \} = k\,   \right.  \mid  \\
&  ~~~~~~~~~ \left. \exists t \leq \tau^{X^1}_0 \, \mbox{ s.t. } R^{X^1}(t)=1 \mbox{ and } S^{X^1}(t) \leq X^1 } \\
&=\lim\limits_{y \rightarrow \infty}  \mathcal{P}^{\nu} {\huge (}  \max\{ x \in \mathbb{N} \, \mbox{ s.t. } \, S^{y}(t)=x,    \, \,  \tilde{\tau}^{y} \leq \exists t < \tau_0^{y} \} =k \mid  \\
& ~~~~~~~~~  \exists t \leq \tau^y_0 \, \mbox{ s.t. } R^y(t)=1 \mbox{ and } S^y(t) \leq X^1 {  )},
\end{align*}
where $\tau^y_0$ is the hitting time of the origin for the random walk $S^y$,
$\tilde{\tau}^y_0 = \max\{t \leq \tau^y_0\, : \, R(t)=1 \}$ is the last time an instruction sleep has been found before hitting
the barrier $A^0$ and the last equality follows 
from the Markov property.
Instead, if no instructions sleep have been found in $[0, X^1]$, we sample $Y^1$ as,
\begin{align*}
\label{eq:sampleY1}
\mathcal{P}^{\nu} \left( Y^1=k \, | \, Y^1 > X^1 \right)   &=\lim\limits_{y \rightarrow \infty}  \mathcal{P}^{\nu} \left( \max\{   x \in \mathbb{N} \, \mbox{ s.t. }  \, 
S^{y}(t)=x  \mbox{ for some } t \mbox{ s.t. } \tilde{\tau}^{y} \leq  t < \tau_0^{y}  \} = k\, |  \right. \\
& \left. ~~~~~~~~~ \nexists t \leq \tau^y_0 \, \mbox{ s.t. } R(t)=1 \mbox{ and } S^y(t) \leq X^1 \right).
\end{align*} 
Thus, for any $k \in \mathbb{N}$,
$$\mathcal{P}^{\nu} \left(Y^1=k  \right) = \lim\limits_{y \rightarrow \infty}  \mathcal{P}^{\nu}\left(\max\{   x \in \mathbb{N} \, \mbox{ s.t. }  \, 
S^{y}(t)=x  \mbox{ for some } t \mbox{ s.t. } \tilde{\tau}^{y} \leq  t < \tau_0^{y}  \} = k\  \right).$$
By symmetry, $Y^1$ is distributed as maximum of $\{S^0(0), S^0(1), \ldots S^0(G)\}$,
where the random walk $S^0$ is conditioned to be positive at all times $t \geq 1$
and  $G$ follows a geometric distribution with parameter $\frac{1 + \lambda}{\lambda}$.
Thus, if $q=0$ then $E[Y^1] = \frac{1 + \lambda}{\lambda}$, whereas
if $q \in (0, \frac{1}{2}]$ then $E[Y^1]  < \frac{1 + \lambda}{\lambda}$.

The proof proceeds now the same as in~\cite[Proof of Theorem 2]{Rolla}. Namely, 
there is a sequence of i.i.d. variables $Y^1$, $Y^2$, $Y^3$, $\ldots$
with the property that the $n$-th step is successful if and only if the previous steps are successful and 
$A^{k-1} + Y^k < X^k$, in which case $A^k = A^{k-1} + Y^k$.
The algorithm succeeds with positive probability if $E[Y^1] <  \frac{1}{\mu}$.
By defining $B(\lambda, q) := \frac{1}{E[Y^1]}$ and by recalling the above-mentioned properties
of  $E[Y^1]$, the proof of the theorem follows.

In particular, by using standard probability tools, one can prove that for any $\lambda \in (0, \infty)$,
$B(\lambda, q)$ is strictly increasing with respect to $q$ in $[0, \frac{1}{2})$
and can derive its analytical expression, which is plotted in  Figure \ref{fig:plot} and \ref{fig:coparison} 
for some values of $\lambda$ and $q$.
\end{proof}

\section{Concluding remarks}
We shall end this article with few comments related to our work.
First of all, our results show that in the case of biased jump distribution, 
by ``stabilizing'' the interval $[-L,L]$, 
the expected number of visits at the origin
is at least linear in $L$ for any $\mu > \mu_1$,
where $\mu_1$ is some number $\mu_1 \geq \mu_c$.
On the other hand, 
such a number is bounded from above by the number of visits  
in the case of no interaction ($\lambda=0$),  which is linear in $L$ 
for any $\mu \in (0, \infty)$.
Hence, it is reasonable to conjecture that
${E}^{\nu}[ m_{[-L,L], \eta, \tau}(0)] = O(L)$ for any $\mu>\mu_c$.

The question whether $\mu_c <1$ has  received
considerable attention recently. In their recent article~\cite{Rolla_Tournier},
Rolla and Tournier consider ARW with biased jump distribution on $\mathbb{Z}^d$ 
and they prove that
$\mu_c(\lambda) \rightarrow 0$ as $\lambda \rightarrow 0$
even when $d \geq 2$.
Concerning the case of unbiased jumps,
the question whether $\mu_c < 1$ for any $\lambda$ is still open in wide generality.
The only positive answer to such a question
has been provided by Stauffer and Taggi~\cite{Stauffer}
on graphs where the random walk has a positive speed. 
The simpler question of $\mu_c < 1$ for $\lambda$ small enough
has been positively answered by Basu, Kanguly and Hoffman~\cite{Basu} on 
$\mathbb{Z}$ and by Stauffer and Taggi~\cite{Stauffer} on all transient graphs.
Remarkably, even such a simpler question remains open for $\mathbb{Z}^2$.

\section*{Acknowledgements}
The author is grateful to Artem Sapozhnikov for illuminating discussions,
for suggesting an argument employed in the proof of Theorem \ref{theo:maintheo2},
and for pointing out a mistake in the proof of Theorem \ref{theo:depend}.
The author thanks Vladas Sidoravicius for inviting him to IMPA and for inspiring discussions
and Augusto Teixeira for fruitful discussions concerning this work.
The author is grateful to Leonardo Rolla for showing how to simplify the proof of 
Theorem \ref{theo:maintheo1}
and for very useful comments. 
The author thanks the referee, as he helped to simplify the exposition significantly.

\end{document}